\documentclass[a4paper]{article}

\usepackage{geometry}
\usepackage{amsmath}                 
\usepackage{amssymb}                 
\usepackage{amsbsy}
\usepackage{amsthm}
\usepackage{amsfonts}

\usepackage{graphicx}
\usepackage{subfigure}
\usepackage{epstopdf,cite}

\usepackage{hyperref}

\newtheorem{theorem}{Theorem}

\newtheorem{property}{Property}

\begin{document}

\title{Epidemic spread in interconnected directed networks}

\author{Junbo Jia$^1$, Zhen Jin$^2$, Xinchu Fu$^{1,}$\thanks{Corresponding author. Tel: +86-21-66132664; Fax: +86-21-66133292; Email address: xcfu@shu.edu.cn} \\ \\
$^1${\small Department of Mathematics, Shanghai University, Shanghai, 200444, China}\\
$^2${\small Complex Systems Research Center, Shanxi University, Taiyuan, Shanxi 030051, China} }



\date{(\today)}


\maketitle

\begin{abstract}
\noindent In the real world, many complex systems interact with other systems. In
addition, the intra- or inter-systems for the spread of information about infectious
diseases and the transmission of infectious diseases are often not random, but with
direction. Hence, in this paper, we build epidemic model based on an interconnected
directed network, which can be considered as the generalization of undirected
networks and bipartite networks. By using the mean-field approach, we establish the
Susceptible-Infectious-Susceptible model on this network. We theoretically analyze
the model, and obtain the basic reproduction number, which is also the generalization
of the critical number corresponding to undirected or bipartite networks. And we
prove the global stability of disease-free and endemic equilibria via the basic
reproduction number as a forward bifurcation parameter. We also give a condition
for epidemic prevalence only on a single subnetwork. Furthermore, we carry out
numerical simulations, and find that the independence between each node's in- and
out-degrees greatly reduce the impact of the network's topological structure on
disease spread.

\medskip

\noindent \textbf{Key words:} Epidemic transmission; interconnected directed network;
basic reproduction number; global stability
\end{abstract}

\section{Introduction}
\label{sec-1}

Many complex systems in the real world can be described by complex networks
~\cite{Dorogovtsev2003}, such as the Internet, WWW of communication system, aviation
networks, railway networks, the metabolic network, gene regulation networks of
biological individuals, friend networks, Facebook or other social networks, etc.
It is no exaggeration to say that networks are everywhere. Some properties of an
actual system can be reflected by the topological structure of corresponding network.
For the 'Six Degree of Separation' theory (also known as the small-world phenomenon)
and the heterogeneity of the number of friends, it can be modelled by small-world
networks~\cite{Watts1998} and scale-free networks~\cite{Barabase1999}, respectively.

There are growing indications that many of real world networks interact with others
~\cite{Liu2016}. For example, in the transportation among cities, there are not
only aviation networks, but also railway networks and road networks. For some
zoonotic diseases (like aftosa, rabies, avian influenza), a human contact network
and an animal network, on which infection relies on, can be considered as whole an
interconnected network. In this paper, we study epidemic spreading dynamics in an
interconnected network, where the nodes in one subnetwork are different from ones
in the other one. By the way, here the interconnected network considered is different
from a multiplex network where all subnetworks may share the same nodes.

According to the directionality of edges, networks can be classified into undirected
networks, directed networks and semi-directed networks~\cite{Zhang2013}. Most of
the previous epidemic models are based on undirected networks~\cite{Pastor2001}.
However, due to the directionality of edges or epidemic spread, it is also suitable
to consider epidemic models based on a directed network. In this network, a susceptible
node receives pathogen only via incoming edges, while an infected node send pathogen
out only via outgoing edges.

In this paper, by using the mean-field approach~\cite{Zhu2015}, we build a
susceptible-infectious-susceptible (SIS) model in an interconnected directed network.
This model can be used to study sexually transmitted diseases, zoonosis, etc. This
foundational network can be seen as a generalization of undirected networks and
bipartite networks. In a special case, if for each directed edge, say $a_{ij}$
(representing directed edge from node $i$ point to node $j$), there exactly exist
a directional opposite edge, $a_{ji}$, then this foundational network can be seen
as an undirected network (or bidirectional network). Alternatively, if there is
only inter-edges  between two subnetworks, without intra-edges within each subnetwork,
then the based network can be seen as a bipartite directed network.

This paper is organized as follows.  In Section~\ref{sec-2}, we establish the SIS
model in an interconnected directed network. In Section~\ref{sec-3} we give a
theoretical analysis with this model, and prove the global stability of disease-free
and endemic equilibria via the basic reproduction number as a forward bifurcation
parameter. Besides, we also give a condition for epidemic prevalence only on a single
subnetwork. In Section~\ref{sec-4}, we perform some numerical simulations to illustrate
and complement our theoretical results. Finally, we summarize some conclusions and
give further discussions in Section~\ref{sec-5}.

\section{SIS model on an interconnected directed network}
\label{sec-2}

This section consists of two subsections, in Section~\ref{subsec-2-1}, we introduce
an interconnected directed network, and give some notations and properties related
to this network. In Section~\ref{subsec-2-2}, we establish an SIS model based on
this network.

\subsection{Interconnected directed networks}
\label{subsec-2-1}

The network we considered here is an interconnected directed network. As shown in
Figure~\ref{fig-1}, this network is composed of two directed subnetworks, subnetwork
$A$ and subnetwork $B$, which are interconnected. The nodes in subnetwork $A$ are
different from the ones in subnetwork  $B$. The nodes of subnetwork $A$ are belong
to some type, while the nodes of subnetwork $B$ are another type. In terms of a
human and animal contact network, the contact network composed of people is regarded
as subnetwork $A$, and the contact network composed of animals is regarded as
subnetwork $B$. Besides, there also exist contacts between two subnetworks.

\begin{figure}[!htbp]                 
  \centering
  \includegraphics[width=7cm]{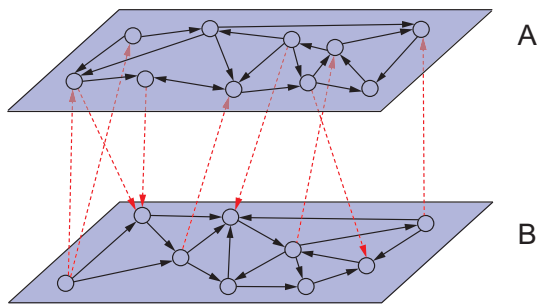}\\
  \caption{Schematic diagram of an interconnected directed network}\label{fig-1}
\end{figure}

Denote by $N$ the total number of nodes in the network, and denote by $N^A$ and
$N^B$ the nodes numbers in subnetworks $A$ and $B$, respectively. So, we have
\begin{equation}\label{eq-1}
  N=N^A+N^B.
\end{equation}

In the network, the connections intra- and inter-subnetworks are directed edges.
For any node in the network, according to the edge direction and the subnetwork
which this node is connected, the attached directed edges can be classified into
the following four types:
\begin{itemize}
  \item Type 1: out-edges pointing to the nodes in subnetwork $A$;
  \item Type 2: in-edges sent out by the nodes in subnetwork $A$;
  \item Type 3: out-edges pointing to the nodes of subnetwork $B$;
  \item Type 4: in-edges sent out by the nodes of subnetwork $B$.
\end{itemize}
Hence, for a node whose number of this four types of directed edges are $k_1, k_2, k_3$
and $k_4$, we say its joint degree is $(k_1,k_2,k_3,k_4)$. We let $N_{k_1, k_2, k_3,k_4}^X$
denote the number of nodes with joint degrees $(k_1, k_2, k_3,k_4)$ in subnetwork
$X$, where the mark $X$ represents $A$ or $B$. We let $n_{1a}$ represent the maximum
of first component in degree $(k_1,k_2,k_3,k_4)$ in subnetwork $A$, and the similar
as $n_{2a}, n_{3a}$, etc. The main notations are shown in Table \ref{tab-1}.

\begin{table}[!htbp]                  
  \centering
  \begin{tabular}{c|l}
\hline
\hline
\textbf{ Notation }   &                                                             
\textbf{ Meaning } ($X$ represent $A$ or $B$)  \\                                   
\hline
\hline
$N_{k_1,k_2,k_3,k_4}^X$   &                                                         
\parbox[1]{10cm}{Number of nodes in $X$ with joint degree $(k_1,k_2,k_3,k_4)$,
$k_1,k_2$ are out-degree and in-degree, respectively, attach to subnetwork $A$,
$k_3,k_4$ are out-degree and in-degree attach to subnetwork $B$}   \\               
\hline
$S_{k_1,k_2,k_3,k_4}^X$ (or $I_{k_1,k_2,k_3,k_4}^X$)   &                            
\parbox[1]{10cm}{Number of susceptible (or infected) nodes in subnetwork $X$ with
joint degree $(k_1,k_2,k_3,k_4)$}   \\                                              
\hline
$s_{k_1,k_2,k_3,k_4}^X$ (or $\rho_{k_1,k_2,k_3,k_4}^X$)   &                         
\parbox[1]{10cm}{Relative density of susceptible (or infected) nodes in subnetwork
$X$ with joint degree $(k_1,k_2,k_3,k_4)$}   \\                                     
\hline
$n_{1a}$(or $n_{2a}$, $n_{3a}$, $n_{4a}$)   &                                       
Maximum degree of $k_1$ (or $k_2$, $k_3$, $k_4$) of nodes in $A$    \\              
\hline
$n_{1b}$(or $n_{2b}$, $n_{3b}$, $n_{4b}$)   &                                       
Maximum degree of $k_1$ (or $k_2$, $k_3$, $k_4$) of nodes in $B$    \\              
\hline
$P_X(k_1,k_2,k_3,k_4)$     &                                                        
\parbox[1]{10cm}{Probability of choosing a random node in $X$ with joint degree
($k_1,k_2,k_3,k_4$)}    \\                                                          
\hline
\hline
\end{tabular}
  \caption{The meaning of main notations.}\label{tab-1}
\end{table}

Let $\Omega$ denote the set of subscripts $\{1,2,3,4\}$, then, for each subnetwork,
the node number $N_{k_1, k_2, k_3,k_4}$ satisfies
\begin{equation}\label{eq-2}
  N^X=\sum_{k_i,i\in \Omega}N_{k_1,k_2,k_3,k_4}^X,
  \qquad X\text{ represent }A \,\text{or}\, B, ~~\text{the same below.}
\end{equation}

On the network, degree distribution $P(k)$ is one of the most fundamental characteristic
quantities. It is defined to be the fraction of nodes in the network with degree $k$,
i.e., $N_k/ N$. Here the degree of node has four components, so we consider the
joint degree distribution. For subnetwork $A$ and $B$, the joint degree distribution
$P_X(k_1,k_2,k_3,k_4)$ is defined respectively as
\begin{equation}\label{eq-3}
  P_X(k_1,k_2,k_3,k_4)=\frac{N_{k_1,k_2,k_3,k_4}^X}{N^X},
\end{equation}
and the	marginal distributions for subnetwork $A$ or $B$ are
\begin{equation}\label{eq-4}
  P_{X}(k_1,\cdot,\cdot,\cdot)=\sum_{k_2,k_3,k_4}P_X(k_1,k_2,k_3,k_4),
\end{equation}
\begin{equation}\label{eq-5}
  P_{X}(\cdot,k_2,\cdot,\cdot)=\sum_{k_1,k_3,k_4}P_X(k_1,k_2,k_3,k_4),
\end{equation}
\begin{equation}\label{eq-6}
  P_{X}(\cdot,\cdot,k_3,\cdot)=\sum_{k_1,k_2,k_4}P_X(k_1,k_2,k_3,k_4),
\end{equation}
\begin{equation}\label{eq-7}
  P_{X}(\cdot,\cdot,\cdot,k_4)=\sum_{k_1,k_2,k_3}P_X(k_1,k_2,k_3,k_4),
\end{equation}
where, $X$ represent $A$ or $B$. Simultaneously, it is easy to verify these marginal
distributions satisfy the normalization condition. And if the joint degree is
independent, then we have
\begin{equation}\label{eq-8}
  P_X(k_1,k_2,k_3,k_4)=P_X(k_1,\cdot,\cdot,\cdot)P_X(\cdot,k_2,\cdot,\cdot)
  P_X(\cdot,\cdot,k_3,\cdot)P_X(\cdot,\cdot,\cdot,k_4).
\end{equation}

Next the mean degree($\phi=1$) and the second moments about zero ($\phi=2$) for
degree are
\[ \langle k_1^\phi\rangle_{a}=\sum_{k_1}k_1^\phi P_A(k_1,\cdot,\cdot,\cdot), \hspace{1.5cm}
\langle k_2^\phi\rangle_{a}=\sum_{k_2}k_2^\phi P_A(\cdot,k_2,\cdot,\cdot), \]
\[ \langle k_3^\phi\rangle_{a}=\sum_{k_3}k_3^\phi P_A(\cdot,\cdot,k_3,\cdot), \hspace{1.5cm}
\langle k_4^\phi\rangle_{a}=\sum_{k_4}k_4^\phi P_A(\cdot,\cdot,\cdot,k_4), \]
\[ \langle k_1^\phi\rangle_{b}=\sum_{k_1}k_1^\phi P_B(k_1,\cdot,\cdot,\cdot), \hspace{1.5cm}
\langle k_2^\phi\rangle_{b}=\sum_{k_2}k_2^\phi P_B(\cdot,k_2,\cdot,\cdot), \]
\[ \langle k_3^\phi\rangle_{b}=\sum_{k_3}k_3^\phi P_B(\cdot,\cdot,k_3,\cdot), \hspace{1.5cm}
\langle k_4^\phi\rangle_{b}=\sum_{k_4}k_4^\phi P_B(\cdot,\cdot,\cdot,k_4). \]

Besides, by the joint degree distribution $P_A(k_1,k_2,k_3,k_4)$, we can obtain
the mixture distribution $P_A(k_1,k_2,\cdot,\cdot)$, as well as mixture moments
$\langle k_1k_2\rangle_{a}$, which satisfies
\[ \langle k_1k_2\rangle_{a}=\sum_{k_1,k_2}k_1k_2 P_A(k_1,k_2 ,\cdot,\cdot). \]
Similarly, we can get other mixture moments $\langle k_2k_3\rangle_{a}$,
$\langle k_1k_2\rangle_{b}$, etc.

Note that there exist relations that the total number of out-edges in $A$ pointing
to $B$ is equal to the ones of in-edges in $B$ coming from $A$, that is,
\[ N^A\langle k_3\rangle_{a}=N^B\langle k_2\rangle_{b}, \]
and, that the total number of in-edges in $A$ coming from $B$ is equal to the number
of out-edges in $B$ pointing to $A$, i.e.,
\[ N^A\langle k_4\rangle_{a}=N^B\langle k_1\rangle_{b}. \]
For subnetworks $A$ and $B$, we also have
\[ \langle k_1\rangle_{a}=\langle k_2\rangle_{a}, \hspace{1.5cm} \langle k_3\rangle_{b}=\langle k_4\rangle_{b}. \]

\subsection{The SIS model on an interconnected directed network}
\label{subsec-2-2}

Now we study the epidemic spread in the interconnected directed network defined
above. Assume that each node must be in an alternative states, susceptible ($S$)
or infected ($I$). We let $I_{k_1,k_2,k_3,k_4}^X$, $X$ represent $A$ or $B$, denote
the number of infected nodes with degree $(k_1,k_2,k_3,k_4)$ in subnetwork $X$,
and $S_{k_1,k_2,k_3,k_4}^X$ denote the number of susceptible nodes. So we have the
following relationships
\begin{equation}\label{eq-9}
S_{k_1,k_2,k_3,k_4}^X+I_{k_1,k_2,k_3,k_4}^X= N_{k_1,k_2,k_3,k_4}^X,
\end{equation}
\[ \sum_{k_i,i\in\Omega} S_{k_1,k_2,k_3,k_4}^X = S^X, \]
\[ \sum_{k_i,i\in\Omega} I_{k_1,k_2,k_3,k_4}^X = I^X. \]

A susceptible node $S$ in subnetwork $A$ can be infected, becoming $I$, only via
in-edges coming from $I$ in subnetwork $A$ or $B$, and the infection rates are
$\lambda_a^a$ and $\lambda_b^a$, respectively. Similarly, a node $S$ in subnetwork
$B$ can be infected only via in-edges coming from $I$ in $A$ or $B$, and the infection
rates are $\lambda_a^b$ and $\lambda_b^b$, respectively. For nodes $I$ whether in
subnetwork $A$ or $B$, they can infect the nodes $S$ in both subnetworks, during
their course of disease. The recovery rate of infected nodes in $A$ and $B$ are
$\mu_a$ and $\mu_b$, respectively.

We assume there exists no degree correlation for any pair of nodes, so the SIS
model is built as follows:
\begin{equation}\label{eq-10}
  \begin{cases}
    \frac{\mathrm{d}I_{k_1,k_2,k_3,k_4}^A(t)}{\mathrm{d}t}=\lambda_a^ak_2S_{k_1,k_2,k_3,k_4}^A(t)\widetilde{\Theta}_a^a(t)
        +\lambda_b^ak_4S_{k_1,k_2,k_3,k_4}^A(t)\widetilde{\Theta}_b^a(t)-\mu_aI_{k_1,k_2,k_3,k_4}^A(t),\\
    \frac{\mathrm{d}I_{k_1,k_2,k_3,k_4}^B(t)}{\mathrm{d}t}=\lambda_a^bk_2S_{k_1,k_2,k_3,k_4}^B(t)\widetilde{\Theta}_a^b(t)
        +\lambda_b^bk_4S_{k_1,k_2,k_3,k_4}^B(t)\widetilde{\Theta}_b^b(t)-\mu_bI_{k_1,k_2,k_3,k_4}^B(t),\\
  \end{cases}
\end{equation}
where, $\widetilde{\Theta}_a^a(t)$ (or $\widetilde{\Theta}_b^a(t)$) represent the
proportion of directed edges coming from infected nodes in total directed edges
coming from nodes in $A$ (or $B$) and pointing to nodes in $A$, $\widetilde{\Theta}_a^b(t)$
(or $\widetilde{\Theta}_b^b(t)$) represent the proportion of directed edges coming
from infected nodes in total directed edges coming from nodes in $A$ (or $B$) and
pointing to nodes in $B$. They are defined as follows:
\[ \widetilde{\Theta}_a^a(t)=\frac{\sum\limits_{l_i,i\in \Omega}l_1I_{l_1,l_2,l_3,l_4}^A}{\sum\limits_{l_i,i\in \Omega}l_1N_{l_1,l_2,l_3,l_4}^A},\hspace{1.5cm}
\widetilde{\Theta}_b^a(t)=\frac{\sum\limits_{l_i,i\in \Omega}l_1I_{l_1,l_2,l_3,l_4}^B}{\sum\limits_{l_i,i\in \Omega}l_1N_{l_1,l_2,l_3,l_4}^B},\]

\[\widetilde{\Theta}_a^b(t)=\frac{\sum\limits_{l_i,i\in \Omega}l_3I_{l_1,l_2,l_3,l_4}^A}{\sum\limits_{l_i,i\in \Omega}l_3N_{l_1,l_2,l_3,l_4}^A},\hspace{1.5cm}
\widetilde{\Theta}_b^b(t)=\frac{\sum\limits_{l_i,i\in \Omega}l_3I_{l_1,l_2,l_3,l_4}^B}{\sum\limits_{l_i,i\in \Omega}l_3N_{l_1,l_2,l_3,l_4}^B}.\]

To simplify this model, we consider the relative densities
$s_{k_1,k_2,k_3,k_4}^X=S_{k_1,k_2,k_3,k_4}^X /N_{k_1,k_2,k_3,k_4}^X$,
$\rho_{k_1,k_2,k_3,k_4}^X=I_{k_1,k_2,k_3,k_4}^X /N_{k_1,k_2,k_3,k_4}^X$,
where $X$ all denote $A$ or $B$ in each equation. By (\ref{eq-2-1}), we have
$s_{k_1,k_2,k_3,k_4}^X+\rho_{k_1,k_2,k_3,k_4}^X= 1$. Hence the model (\ref{eq-10})
becomes

\begin{equation}\label{eq-11}
  \begin{cases}
    \frac{\mathrm{d}\rho_{k_1,k_2,k_3,k_4}^A(t)}{\mathrm{d}t}=\lambda_a^ak_2\left(1-\rho_{k_1,k_2,k_3,k_4}^A(t)\right)\Theta_a^a(t)
        +\lambda_b^ak_4\left(1-\rho_{k_1,k_2,k_3,k_4}^A(t)\right)\Theta_b^a(t)-\mu_a\rho_{k_1,k_2,k_3,k_4}^A(t), \\
    \frac{\mathrm{d}\rho_{k_1,k_2,k_3,k_4}^B(t)}{\mathrm{d}t}=\lambda_a^bk_2\left(1-\rho_{k_1,k_2,k_3,k_4}^B(t)\right)\Theta_a^b(t)
        +\lambda_b^bk_4\left(1-\rho_{k_1,k_2,k_3,k_4}^B(t)\right)\Theta_b^b(t)-\mu_b\rho_{k_1,k_2,k_3,k_4}^B(t), \\
  \end{cases}
\end{equation}
where
\[ \Theta_a^a(t)=\frac{\sum\limits_{l_i,i\in \Omega}l_1I_{l_1,l_2,l_3,l_4}^A}{\sum\limits_{l_i,i\in \Omega}l_1N_{l_1,l_2,l_3,l_4}^A}
                =\frac{\sum\limits_{l_i,i\in \Omega}l_1P_A(l_1,l_2,l_3,l_4)\rho_{l_1,l_2,l_3,l_4}^A}{\sum\limits_{l_i,i\in \Omega}l_1P_A(l_1,l_2,l_3,l_4)}
                =\frac{1}{\langle k_1\rangle_{a}}\sum\limits_{l_i,i\in \Omega}l_1P_A(l_1,l_2,l_3,l_4)\rho_{l_1,l_2,l_3,l_4}^A, \]

\[ \Theta_b^a(t)=\frac{\sum\limits_{l_i,i\in \Omega}l_1I_{l_1,l_2,l_3,l_4}^B}{\sum\limits_{l_i,i\in \Omega}l_1N_{l_1,l_2,l_3,l_4}^B}
                =\frac{\sum\limits_{l_i,i\in \Omega}l_1P_B(l_1,l_2,l_3,l_4)\rho_{l_1,l_2,l_3,l_4}^B}{\sum\limits_{l_i,i\in \Omega}l_1P_B(l_1,l_2,l_3,l_4)}
                =\frac{1}{\langle k_1\rangle_{b}}\sum\limits_{l_i,i\in \Omega}l_1P_B(l_1,l_2,l_3,l_4)\rho_{l_1,l_2,l_3,l_4}^B, \]

\[ \Theta_a^b(t)=\frac{\sum\limits_{l_i,i\in \Omega}l_3I_{l_1,l_2,l_3,l_4}^A}{\sum\limits_{l_i,i\in \Omega}l_3N_{l_1,l_2,l_3,l_4}^A}
                =\frac{\sum\limits_{l_i,i\in \Omega}l_3P_A(l_1,l_2,l_3,l_4)\rho_{l_1,l_2,l_3,l_4}^A}{\sum\limits_{l_i,i\in \Omega}l_3P_A(l_1,l_2,l_3,l_4)}
                =\frac{1}{\langle k_3\rangle_{a}}\sum\limits_{l_i,i\in \Omega}l_3P_A(l_1,l_2,l_3,l_4)\rho_{l_1,l_2,l_3,l_4}^A, \]

\[ \Theta_b^b(t)=\frac{\sum\limits_{l_i,i\in \Omega}l_3I_{l_1,l_2,l_3,l_4}^B}{\sum\limits_{l_i,i\in \Omega}l_3N_{l_1,l_2,l_3,l_4}^B}
                =\frac{\sum\limits_{l_i,i\in \Omega}l_3P_B(l_1,l_2,l_3,l_4)\rho_{l_1,l_2,l_3,l_4}^B}{\sum\limits_{l_i,i\in \Omega}l_3P_A(l_1,l_2,l_3,l_4)}
                =\frac{1}{\langle k_3\rangle_{b}}\sum\limits_{l_i,i\in \Omega}l_3P_B(l_1,l_2,l_3,l_4)\rho_{l_1,l_2,l_3,l_4}^B. \]

Note that the model~(\ref{eq-11}) generalizes SIS models in the following special
networks:
\begin{enumerate}
  \item Undirected networks: For any directed edge in the network, such as $a_{ij}$
    representing that it points to node $j$ by node $i$, if there exactly
    exists an directional opposite directed edges $a_{ji}$, then we say
    that this network can be considered as an undirected network, or as a
    bidirectional network;
  \item Single-layer networks: If we remove the subnetwork $B$ and let
    some related degrees be zeros, such as $n_{3a}, n_{4a}, n_{1b}, n_{2b}, n_{3b}, n_{4b}$,
    then this network will be a single-layer network;
  \item Bipartite networks: For some interconnected directed networks,
  if there exists no intra-edges within each subnetwork, and
    there exist only inter-edges between two subnetworks, then this special
    network is a bipartite network.
\end{enumerate}


\section{Mathematical analysis}
\label{sec-3}

In mathematical epidemiology, the basic reproduction number $R_0$ gives the number
of secondary cases one infectious individual will produce in a population consisting
only of susceptible individuals. Mathematically, the basic reproduction number
$R_0$ plays the role of a threshold value for the epidemic dynamics. If $R_0 > 1$,
the disease will break out, otherwise, the number of infected individuals gradually
declines to zero, and the disease disappears from the population.

Then we will calculate the basic reproduction number $R_0$ for model~(\ref{eq-11}),
so as to study the global dynamical behaviors. For simplicity, we denote
$\mathbf{x}=(x_1, x_2, x_3,\cdots,x_n)=(\rho_{0,0,0,0}^A, \rho_{0,0,0,1}^A,\cdots,\rho_{n_{1b}, n_{2b}, n_{3b}, n_{4b}}^B)$,
where $n=\prod\limits_{i\in \Omega ,j=a,b}(n_{ij}+1)$, and
$\mathbf{f}=(f_1, f_2,\cdots,f_n)$ be the functions of the right-hand side of
(\ref{eq-11}). Then the model (\ref{eq-11}) can been rewritten as
\begin{equation}\label{eq-12}
\frac{\mathrm{d}\mathbf{x}(t)}{\mathrm{d}t}=\mathbf{f(x}(t)).
\end{equation}
It is obvious that there always exists a disease-free equilibrium (DFE)
$E_0=(0,0,\cdots,0)$ in the model (\ref{eq-11}) with $x_i=0, i=1, 2, \cdots,n$.

The basic reproduction number $R_0$ is a crucial parameter in epidemic dynamics.
Next we calculate $R_0$ by estimating the spectral radius of the generation matrix
$\Gamma$~\cite{Driessche2002}. Based on this method, one has $\Gamma=FV^{-1}$,
where $F$ is the rate of new occurring infections and $V$ is the rate of transferring
individuals out of the original group, and the matrix $\Gamma$ can be expressed as
\begin{equation}\label{eq-13}
\begin{aligned}
 \Gamma&=\left[
\begin{matrix}
\lambda_a^a\sum\limits_{l_i,i\in \Omega}\frac{l_1l_2P_A(l_1,l_2,l_3,l_4)}{\mu_a\langle k_1\rangle_{a}}  &
\lambda_a^a\sum\limits_{l_i,i\in \Omega}\frac{l_1l_4P_A(l_1,l_2,l_3,l_4)}{\mu_a\langle k_1\rangle_{a}}  &
0  &
0  \\
0  &
0  &
\lambda_b^a\sum\limits_{l_i,i\in \Omega}\frac{l_1l_2P_B(l_1,l_2,l_3,l_4)}{\mu_b\langle k_1\rangle_{b}}  &
\lambda_b^a\sum\limits_{l_i,i\in \Omega}\frac{l_1l_4P_B(l_1,l_2,l_3,l_4)}{\mu_b\langle k_1\rangle_{b}}  \\
\lambda_a^b\sum\limits_{l_i,i\in \Omega}\frac{l_2l_3P_A(l_1,l_2,l_3,l_4)}{\mu_a\langle k_3\rangle_{a}}  &
\lambda_a^b\sum\limits_{l_i,i\in \Omega}\frac{l_3l_4P_A(l_1,l_2,l_3,l_4)}{\mu_a\langle k_3\rangle_{a}}  &
0  &
0  \\
0  &
0  &
\lambda_b^b\sum\limits_{l_i,i\in \Omega}\frac{l_2l_3P_B(l_1,l_2,l_3,l_4)}{\mu_b\langle k_3\rangle_{b}}  &
\lambda_b^b\sum\limits_{l_i,i\in \Omega}\frac{l_3l_4P_B(l_1,l_2,l_3,l_4)}{\mu_b\langle k_3\rangle_{b}}
\end{matrix}
\right]\\
&=\left[
\begin{matrix}
\frac{\lambda_a^a\langle k_1k_2\rangle_a}{\mu_a\langle k_1\rangle_{a}}  &
\frac{\lambda_a^a\langle k_1k_4\rangle_a}{\mu_a\langle k_1\rangle_{a}}  &
0   &
0   \\
0   &
0   &
\frac{\lambda_b^a\langle k_1k_2\rangle_b}{\mu_b\langle k_1\rangle_{b}}  &
\frac{\lambda_b^a\langle k_1k_4\rangle_b}{\mu_b\langle k_1\rangle_{b}}  \\
\frac{\lambda_a^b\langle k_2k_3\rangle_a}{\mu_a\langle k_3\rangle_{a}}  &
\frac{\lambda_a^b\langle k_3k_4\rangle_a}{\mu_a\langle k_3\rangle_{a}}  &
0   &
0   \\
0   &
0   &
\frac{\lambda_b^b\langle k_2k_3\rangle_b}{\mu_b\langle k_3\rangle_{b}}  &
\frac{\lambda_b^b\langle k_3k_4\rangle_b}{\mu_b\langle k_3\rangle_{b}}
\end{matrix}
\right].
\end{aligned}
\end{equation}

Below we study two special cases: Case 1 is undirected network,
where for each node their out-degree is equal to their in-degree.
In this sense, undirected network is a special directed network
with correlation between every node's out-degree and in-degree;
However in Case~2, for each node the joint degree is independent,
and there exists no correlation between out-degree and in-degree.

Case 1. If we consider this network as an undirected network, i.e.,
for every directed edge $a_{ij}$ there exactly exists an edge $a_{ji}$,
then every node's out-degrees will be equal to the corresponding
in-degree, respectively, i.e., $k_1=k_2,k_3=k_4$. We also assume $k_1$
(or $k_2$) is independent of $k_3$ ($k_4$). Then we have
\[P_X(l_1,l_2,l_3,l_4)=P_X(l_1,\cdot,\cdot,\cdot)P_X(\cdot,\cdot,l_3,\cdot),\]
and
\[\langle k_1\rangle_{a}=\langle k_2\rangle_{a},\qquad \langle k_3\rangle_{a}=\langle k_4\rangle_{a},\]
\[\langle k_1\rangle_{b}=\langle k_2\rangle_{b},\qquad \langle k_3\rangle_{b}=\langle k_4\rangle_{b}.\]

Substituting the above into (\ref{eq-13}), we obtain
\begin{equation}\label{eq-14}
\begin{aligned}
\Gamma_1 & =\left[
\begin{matrix}
\lambda_a^a\frac{\sum\limits_{l_1}l_1^2P_A(l_1,\cdot,\cdot,\cdot)}{\mu_a\langle k_1\rangle_a}   \hspace{1cm}   \lambda_a^a\frac{\sum\limits_{l_1l_3}l_1P_A(l_1,\cdot,\cdot,\cdot)l_3P_A(\cdot,\cdot,l_3,\cdot)}{\mu_a\langle k_1\rangle_a}  &
0   \hspace{4cm} 0   \\
0   \hspace{4cm} 0   &
\lambda_b^a\frac{\sum\limits_{l_1}l_1^2P_B(l_1,\cdot,\cdot,\cdot)}{\mu_b\langle k_1\rangle_b}   \hspace{1cm}   \lambda_b^a\frac{\sum\limits_{l_1l_3}l_1P_B(l_1,\cdot,\cdot,\cdot)l_3P_B(\cdot,\cdot,l_3,\cdot)}{\mu_b\langle k_1\rangle_b}   \\
\lambda_a^b\frac{\sum\limits_{l_1l_3}l_1P_A(l_1,\cdot,\cdot,\cdot)l_3P_A(\cdot,\cdot,l_3,\cdot)}{\mu_a\langle k_3\rangle_a}   \hspace{1cm}  \lambda_a^b\frac{\sum\limits_{l_3}l_3^2P_A(\cdot,\cdot,l_3,\cdot)}{\mu_a\langle k_3\rangle_a}  &
0   \hspace{4cm}
0   \\
0   \hspace{4cm}
0   &
\lambda_b^b\frac{\sum\limits_{l_1l_3}l_1l_3P_B(l_1,\cdot,\cdot,\cdot)P_B(\cdot,\cdot,l_3,\cdot)}{\mu_b\langle k_3\rangle_b}   \hspace{1cm}   \lambda_b^b\frac{\sum\limits_{l_3}l_3^2P_B(\cdot,\cdot,l_3,\cdot)}{\mu_b\langle k_3\rangle_b}
\end{matrix}
\right]\\
&=\left[
\begin{matrix}
\frac{\lambda_a^a\langle k_1^2\rangle_a}{\mu_a\langle k_1\rangle_a}   &
\frac{\lambda_a^a\langle k_3\rangle_a}{\mu_a}   &
0   &
0   \\
0   &
0   &
\frac{\lambda_b^a\langle k_1^2\rangle_b}{\mu_b\langle k_1\rangle_b}   &
\frac{\lambda_b^a\langle k_3\rangle_b}{\mu_b}   \\
\frac{\lambda_a^b\langle k_1\rangle_a}{\mu_a}   &
\frac{\lambda_a^b\langle k_3^2\rangle_a}{\mu_a\langle k_3\rangle_a}   &
0   &
0   \\
0   &
0   &
\frac{\lambda_b^b\langle k_1\rangle_a}{\mu_b}   &
\frac{\lambda_b^b\langle k_3^2\rangle_b}{\mu_b\langle k_3\rangle_b}
\end{matrix}
\right].
\end{aligned}
\end{equation}


Case~2. If we don't consider the correlation among sub-degrees, then
the joint degree distribution is independent, i.e.,
$P_X(l_1,l_2,l_3,l_4)=P_X(l_1,\cdot,\cdot,\cdot) P_X(\cdot,l_2,\cdot,\cdot)
 P_X(\cdot,\cdot,l_3,\cdot)$ $P_X(\cdot,\cdot,\cdot,l_4) $, then we have
\begin{equation}\label{eq-15}
\Gamma_2=\left[
\begin{matrix}
\frac{\lambda_a^a\langle k_2\rangle_a}{\mu_a}   &
\frac{\lambda_a^a\langle k_4\rangle_a}{\mu_a}   &
0   &
0   \\
0   &
0   &
\frac{\lambda_b^a\langle k_2\rangle_b}{\mu_b}   &
\frac{\lambda_b^a\langle k_4\rangle_b}{\mu_b}   \\
\frac{\lambda_a^b\langle k_2\rangle_a}{\mu_a}   &
\frac{\lambda_a^b\langle k_4\rangle_a}{\mu_a}   &
0   &
0   \\
0   &
0   &
\frac{\lambda_b^b\langle k_2\rangle_b}{\mu_b}   &
\frac{\lambda_b^b\langle k_4\rangle_b}{\mu_b}
\end{matrix}
\right].
\end{equation}

Since $\Gamma$ is an irreducible non-negative matrix, by Perron-Frobenius theorem,
there is a positive real eigenvalue which is equal to the spectral radius $\rho(\Gamma)$.
Denote $s(Df(0)) = \max\{Re \lambda; \det(\lambda \mathbf{I} - D\mathbf{f}(0))=0\}$.
It is shown in~\cite{Diekmann1990} that: $s(Df(0)) \leq 0 \Longleftrightarrow \rho(\Gamma) < 1$,
and $s(Df(0)) > 0 \Longleftrightarrow \rho(\Gamma) > 1$. Hence, the basic reproduction
number of model (\ref{eq-11}) is
\begin{equation}\label{eq-16}
  R_0=\rho(\Gamma).
\end{equation}

We have already known that this interconnected directed network is
a generalized network, and the disease spreading have relationship
with the structure of the network. Hence, the $R_0$ also generalize
the basic reproduction numbers for the special cases.
If the network becomes some special cases, such as a single layer network,
or a bipartite network, the $R_0$ is also satisfied, but the expression
of $R_0$ will do some deformation correspondingly. Specific cases is
discussed as follows:
\begin{enumerate}
  \item Single layer network. If we delete the subnetwork $B$
  and the all inter-edges between two subnetworks, then only
  subnetwork $A$ remains, becoming a single layer directed
  network. The corresponding generation matrix is a sub-matrix of $\Gamma$ (\ref{eq-13})
  with their upper left $2\times 2$ items. Obviously we have
  $R_0 = \frac{\lambda_a^a\langle k_1k_2\rangle_a}{\mu_a\langle k_1\rangle_{a}}$; Furthermore, if this
  single layer network is undirected, then each node's out-degree is
  equal to its in-degree, therefore, the generation matrix will
  be a sub-matrix of $\Gamma_1$ (\ref{eq-14}) with the upper left $2\times 2$ items,
  and hence the $R_0$ is
  $\frac{\lambda_a^a\langle k_1^2\rangle_{a}}{\mu_a\langle k_1\rangle_{a}}$;
  Alternatively, if each node's out-degree is independent of its in-degree, then
  the generation matrix is a sub-matrix of $\Gamma_2$ (\ref{eq-15}), and the
  $R_0$ is $\frac{\lambda_a^a\langle k_1\rangle_{a}}{\mu_a}$. These results
  are listed in Table~\ref{tab-2}. Besides, if all the infectious rates,
  intra-subnetworks or inter-subnetworks, are equal, and the recovery rate in subnetwork
  $A$ and $B$ are also equal, that is to say, $\lambda_a^a = \lambda_b^a = \lambda_a^b = \lambda_b^b$,
  $\mu_a=\mu_b$, then this network can also be considered as a single layer network.

  \item Bipartite network. If we delete the intra-edges within subnetworks $A$
  and $B$, the interconnected directed network then becomes a bipartite network.
  The generation matrix becomes the sub-matrix of $\Gamma$ (\ref{eq-13}) with
  central $2\times 2$ items. So the $R_0$ is $\sqrt{\frac{\lambda_a^b\langle k_3k_4\rangle_a}{\mu_a\langle k_3\rangle_{a}}\cdot\frac{\lambda_b^a\langle k_1k_2\rangle_b}{\mu_b\langle k_1\rangle_{b}}}$.
  In the cases of networks of undirected and one in which node's out-degree independent of their in-degree,
  the corresponding $R_0$ are $\sqrt{\frac{\lambda_a^b\langle k_3^2\rangle_a}{\mu_a \langle k_3\rangle_a} \cdot \frac{\lambda_b^a\langle k_1^2\rangle_b}{\mu_b \langle k_1\rangle_b}}$ and
  $\sqrt{\frac{\lambda_a^b\langle k_4\rangle_a}{\mu_a}\cdot \frac{\lambda_b^a\langle k_2\rangle_b}{\mu_b}}$,
  respectively.

  \item Interconnected network. For this situation, though the $R_0$ can not be
  explicitly expressed, by Perron-Frobenius theorem we can give the following estimation:
      $$\min_{i,j}\{c_i,r_j\} \leq R_0 \leq \max_{i,j}\{c_i,r_j\}, $$
  where $c_i$ and $r_j$ are the $i^{th}$ row sums and $j^{th}$ column sums on $\Gamma$,
  respectively. If the interconnected network is undirected, then the estimating method is the same.
  If their joint degree is independent, then we have
  \begin{equation}\label{eq-17}
    R_0=\rho(\Gamma_2)=\frac{1}{2} \left(\frac{\lambda_a^a\langle k_2\rangle_a}{\mu_a} + \frac{\lambda_b^b\langle k_4\rangle_b}{\mu_b}+\sqrt{\left(\frac{\lambda_a^a\langle k_2\rangle_a}{\mu_a}-\frac{\lambda_b^b\langle k_4\rangle_b}{\mu_b}\right)^2+4 \frac{\lambda_a^b\langle k_4\rangle_a}{\mu_a}\cdot \frac{\lambda_b^a\langle k_2\rangle_b}{\mu_b}}\right).
  \end{equation}
\end{enumerate}

\renewcommand\arraystretch{2} 
\begin{table}[!htbp]
\centering
\begin{tabular}{c|c|c|c}
\hline
\hline
$R_0$           &                   
Single layer    &                   
Bipartite       &                   
Interconnected  \\                  
\hline
\hline
Directed        &                   
$\frac{\lambda_a^a\langle k_1k_2\rangle_a}{\mu_a\langle k_1\rangle_{a}}$ \cite{Wang2009,Schwartz2002}   &   
$\sqrt{\frac{\lambda_a^b\langle k_3k_4\rangle_a}{\mu_a\langle k_3\rangle_{a}}\cdot\frac{\lambda_b^a\langle k_1k_2\rangle_b}{\mu_b\langle k_1\rangle_{b}}}$   &               
$\rho(\Gamma)$      \\              
\hline
Undirected      &                   
$\frac{\lambda_a^a\langle k_1^2\rangle_{a}}{\mu_a\langle k_1\rangle_{a}}$ \cite{Pastor2001-scale}   &       
$\sqrt{\frac{\lambda_a^b\langle k_3^2\rangle_a}{\mu_a \langle k_3\rangle_a} \cdot \frac{\lambda_b^a\langle k_1^2\rangle_b}{\mu_b \langle k_1\rangle_b}}$ \cite{Jesus2008}    &   
$\rho(\Gamma_1)$                    \\  
\hline
Directed (Independent)              &   
$\frac{\lambda_a^a\langle k_1\rangle_{a}}{\mu_a}$\cite{Tanimoto2011,Wang2009,Schwartz2002}  &               
$\sqrt{\frac{\lambda_a^b\langle k_4\rangle_a}{\mu_a}\cdot \frac{\lambda_b^a\langle k_2\rangle_b}{\mu_b}}$   &   
$\rho(\Gamma_2)$                    \\  
\hline
\hline
\end{tabular}
\caption{$R_0$ for networks with different connection patterns.}\label{tab-2}
\end{table}


For an interconnected directed network, in addition to $R_0$, we denote by $R_0^{single\,X}$
the basic reproduction number of subnetwork $X$, which
is separated from the whole interconnected network. And we also let $R_0^{bipartite}$
represent the basic reproduction number of the separated bipartite network. Then for Case~2,
we have the following property.

\begin{property}
For the model (\ref{eq-11}), if the joint degree distribution is independent,
then $R_0\geq R_0^{single\,X}$ and $R_0\geq R_0^{bipartite}$.
\end{property}

\begin{proof}
If the joint degree distribution is independent, then
\begin{equation*}
  \begin{split}
     R_0=\rho(\Gamma_2)&=\frac{1}{2} \left(\frac{\lambda_a^a\langle k_2\rangle_a}{\mu_a} + \frac{\lambda_b^b\langle k_4\rangle_b}{\mu_b}+\sqrt{\left(\frac{\lambda_a^a\langle k_2\rangle_a}{\mu_a}-\frac{\lambda_b^b\langle k_4\rangle_b}{\mu_b}\right)^2+4 \frac{\lambda_a^b\langle k_4\rangle_a}{\mu_a}\cdot \frac{\lambda_b^a\langle k_2\rangle_b}{\mu_b}}\right)\\
     & \geq \frac{1}{2} \left(\frac{\lambda_a^a\langle k_2\rangle_a}{\mu_a} + \frac{\lambda_b^b\langle k_4\rangle_b}{\mu_b} +
     \left|\frac{\lambda_a^a\langle k_2\rangle_a}{\mu_a}-\frac{\lambda_b^b\langle k_4\rangle_b}{\mu_b}\right| \right) \\
     & \geq \frac{1}{2} \left(\frac{\lambda_a^a\langle k_2\rangle_a}{\mu_a} + \frac{\lambda_b^b\langle k_4\rangle_b}{\mu_b} +
     \frac{\lambda_a^a\langle k_2\rangle_a}{\mu_a}-\frac{\lambda_b^b\langle k_4\rangle_b}{\mu_b} \right)
     = \frac{\lambda_a^a\langle k_2\rangle_a}{\mu_a}
     = R_0^{single\,A}.
  \end{split}
\end{equation*}

Similarly, we have $R_0\geq R_0^{single\,B}$. In addition, we also have
\begin{equation*}
  \begin{split}
     R_0=\rho(\Gamma_2)&=\frac{1}{2} \left(\frac{\lambda_a^a\langle k_2\rangle_a}{\mu_a} + \frac{\lambda_b^b\langle k_4\rangle_b}{\mu_b}+\sqrt{\left(\frac{\lambda_a^a\langle k_2\rangle_a}{\mu_a}-\frac{\lambda_b^b\langle k_4\rangle_b}{\mu_b}\right)^2+4 \frac{\lambda_a^b\langle k_4\rangle_a}{\mu_a}\cdot \frac{\lambda_b^a\langle k_2\rangle_b}{\mu_b}}  \right)\\
     & \geq \frac{1}{2} \left(\frac{\lambda_a^a\langle k_2\rangle_a}{\mu_a} + \frac{\lambda_b^b\langle k_4\rangle_b}{\mu_b} +
     2\sqrt{ \frac{\lambda_a^b\langle k_4\rangle_a}{\mu_a}\cdot \frac{\lambda_b^a\langle k_2\rangle_b}{\mu_b}}  \right) \\
     & \geq \sqrt{ \frac{\lambda_a^b\langle k_4\rangle_a}{\mu_a}\cdot \frac{\lambda_b^a\langle k_2\rangle_b}{\mu_b}}
     = R_0^{bipartite}.
  \end{split}
\end{equation*}
In summary, we have $R_0\geq R_0^{single\,X}$, $X\in\{A,B\}$, and $R_0\geq R_0^{bipartite}$.
The proof is completed.
\end{proof}

We now present the main result about the dynamical behavior of the model as follows.
\begin{theorem}\label{th-1}
For the model (\ref{eq-11}), if $R_0\leq1$, then the DFE $E_0$ is globally asymptotically stable in $\Delta$; While if $R_0 > 1$, model (\ref{eq-11}) admits a unique endemic equilibrium (EE) $E_1$, which is globally asymptotically stable in $\Delta -\{0\}$, where $\Delta=\{(x_1,x_2,\cdots,x_n)\in R_+^n:0\leqslant x_i \leqslant1, i=1,\cdots,n\}$.
\end{theorem}
\begin{proof}
First, we prove that the set $\Delta$ is positive invariant
for model (\ref{eq-11}). We claim that if initial value $x(0)\in \Delta$,
then $x_i(t)\geqslant 0,\forall t>0,~ i=1,\cdots n$.
Otherwise, there would exist a $k_0 \in \{1,\cdots n\}$ and $t_0 > 0$,
such that $x_{k_0}<0$. Since $x_{k_0}<0$ is continuous, by the zero
point theorem, there must exist $t_1$ such that $x_{k_0}(t_1)=0$.
Let $t^*$ be the infimum of all zero points for all functions $x_i(t)$,
i.e., $t^*=\inf \{t>0, x_i(t)=0,i=1,2,\cdots,n\}$. Without loss of generality,
let $x_{k_1}(t^*)=\rho_{k_1,k_2,k_3,k_4}^A(t^*)=0$, then by the definition
of $t^*$, we have $d\rho_{k_1,k_2,k_3,k_4}^A(t^*)/dt\leqslant 0$.
However, by the model (\ref{eq-11}), $d\rho_{k_1,k_2,k_3,k_4}^A(t^*)/ dt = \lambda_a^a k_2\Theta_{11}(t^*)+\lambda_b^ak_4\Theta_{21}(t)>0$ is derived, which contradicts the
proceeding formula. So $x_i(t)\geqslant 0$ holds for all $t>0$ and $i=1,2,\cdots,n$.
Similarly, we can prove $s_{k_1,k_2,k_3,k_4}^A\geqslant 0$ and $s_{k_1,k_2,k_3,k_4}^B\geqslant 0$.
Because of $s_{k_1,k_2,k_3,k_4}^A=1-\rho_{k_1,k_2,k_3,k_4}^A$ and $s_{k_1,k_2,k_3,k_4}^B=1-\rho_{k_1,k_2,k_3,k_4}^B$, we can get $\rho_{k_1,k_2,k_3,k_4}^A\leqslant 1$
and $\rho_{k_1,k_2,k_3,k_4}^B\leqslant 1$. Hence $0 \leqslant x_{i}(t) \leqslant 1$ hold for
all $t>0$ and $i=1,2,\cdots n$. Thus the claim is proven.

Next, we verify that model (\ref{eq-11}) satisfies the conditions
of Corollary 3.2 in \cite{Zhao1996}. It is clear that the function
$f: \Delta\rightarrow R^n$ defined in (\ref{eq-11}) is continuously differentiable.
Direct calculation yields $ \partial f_i/\partial x_j \geqslant 0$ for $x\in \Delta$
and $i \neq j$, which implies that $f$ is cooperative.
Further, we know that $Df=(\partial f_i/\partial x_j)_{1\leqslant i,j\leqslant n}$
is irreducible for every $x\in \Delta$. At the
same time, it is obvious that $f(0)=0$ and $f_i(x)\geqslant 0$ for all
$x\in \Delta$ with $x_i=0,i=1,2,\cdots,n $. Moreover, since for any
$\varepsilon \in (0,1)$ and $x\gg 0, f_i(\varepsilon x)\geqslant \varepsilon f_i(x),
i=1,2,\cdots,n$, we have that $f$ is strictly sublinear in $\Delta$.
By applying Corollary 3.2 in \cite{Zhao1996}, the proof is completed.
\end{proof}

Theorem \ref{th-1} give the condition that epidemic whether or not prevalent on whole
interconnected network, and the next theorem give a condition that epidemic prevalent
only on the single subnetwork.

\begin{theorem}\label{th-2}
If there exists an interconnected directed network, and $R_0>1$, but some subnetwork's
$R_0$ is less than one, for example $R_0^{single\,A}<1$, in addition, this network is one-way~\cite{Wang2016},
i.e., there exist no directed edges from subnetwork $B$ to subnetwork $A$ or $\lambda_b^a=0$.
Then there exist two equilibria for model~(\ref{eq-11}), one is DFE $E_0$ which is a saddle,
and another is boundary EE $E_1$ which is globally
asymptotically stable in $\Delta -\{0\}$.
\end{theorem}

\begin{proof}
First, we prove that the zero is the only equilibrium for $\rho_{k_1,k_2,k_3,k_4}^A(t)$ in $\Delta$.
Since there exist no directed edges from subnetwork $B$ to subnetwork $A$ or $\lambda_b^a=0$,
model (\ref{eq-11}) becomes
\begin{equation}\label{eq-18}
  \begin{cases}
    \frac{\mathrm{d}\rho_{k_1,k_2,k_3,k_4}^A(t)}{\mathrm{d}t}=\lambda_a^ak_2\left(1-\rho_{k_1,k_2,k_3,k_4}^A(t)\right)\Theta_a^a(t)
        -\mu_a\rho_{k_1,k_2,k_3,k_4}^A(t),\\
    \frac{\mathrm{d}\rho_{k_1,k_2,k_3,k_4}^B(t)}{\mathrm{d}t}=\lambda_a^bk_2\left(1-\rho_{k_1,k_2,k_3,k_4}^B(t)\right)\Theta_a^b(t)
        +\lambda_b^bk_4\left(1-\rho_{k_1,k_2,k_3,k_4}^B(t)\right)\Theta_b^b(t)-\mu_b\rho_{k_1,k_2,k_3,k_4}^B(t).\\
  \end{cases}
\end{equation}
From (\ref{eq-18}), we understand the derivative of $\rho_{k_1,k_2,k_3,k_4}^A(t)$
have nothing to do with $\rho_{k_1,k_2,k_3,k_4}^B(t)$. So, the first equation of
(\ref{eq-18}) can be regarded as an SIS model in subnetwork $A$ as a separated
single layer network. Since $R_0^{single\,A}<1$, according to Theorem \ref{th-1},
we obtain that the zero is the only equilibrium for $\rho_{k_1,k_2,k_3,k_4}^A(t)$ in $\Delta$,
and it is globally asymptotically stable. That is to say, $\rho_{k_1,k_2,k_3,k_4}^A(t)$
tends to zero for all joint degree $(k_1,k_2,k_3,k_4)$ in subnetwork $A$.

While for the whole interconnected network, we have $R_0>1$, so model~(\ref{eq-11}) must
have two equilibria, the DFE $E_0$ and the EE $E_1$, and $E_1$ is globally
asymptotically stable in $\Delta -\{0\}$. Then there must exist
some joint degree $(k_1^*,k_2^*,k_3^*,k_4^*)$ in the subnetwork $B$, such that
$\rho_{k_1^*,k_2^*,k_3^*,k_4^*}^B(t)$ does not tend to zero but tends to $E_1$.
So the DFE $E_0$ is a saddle and the EE $E_1$ is on the boundary of
$\Delta -\{0\}$. The proof is completed.
\end{proof}

\section{Numerical analysis}
\label{sec-4}

To illustrate and complement the above theoretical analysis and further explore the disease dynamics,
we perform numerical simulations for model (\ref{eq-11}) on different networks.
Here we mainly study the situation which node's joint degree distribution is independent.

For the basic reproduction number $R_0$, we have known that it is equal to
the spectral radius of $\Gamma$, i.e., $R_0=\rho(\Gamma)$. With regard
to $\Gamma_2$ in (\ref{eq-14}), we calculate its spectral radius with the
given parameters. As shown in Figure~\ref{fig-2},
except the variable presented in the figures, other infectious
rates are fixed to 0.1, mean degree $\langle k_i\rangle_X$ are fixed to 4.
Besides, the recovery rates are set to one, $\mu_a=\mu_b=1$.
We can see that the $R_0$ is increasing with the increase of infectious rate $\lambda$
and mean degree $\langle k\rangle$. And the influence infectious rate
on $R_0$ is the same as the influence of mean degree. In addition, the
inner infectious rates, i.e., $\lambda_a^a$ and $\lambda_b^b$ have greater influence
than the cross infectious rates $\lambda_a^b$ and $\lambda_b^a$ on $R_0$ in the
same parameter value, and the intra- mean
degree $\langle k_2\rangle_a,\langle k_4\rangle_b$ also have greater
influence than inter- mean degrees $\langle k_4\rangle_a$ and $\langle k_2\rangle_b$.

\begin{figure}[!h]
  \centering
  \subfigure{
  \begin{minipage}{7cm}
  (a)
  \centering
  \includegraphics[width=6cm]{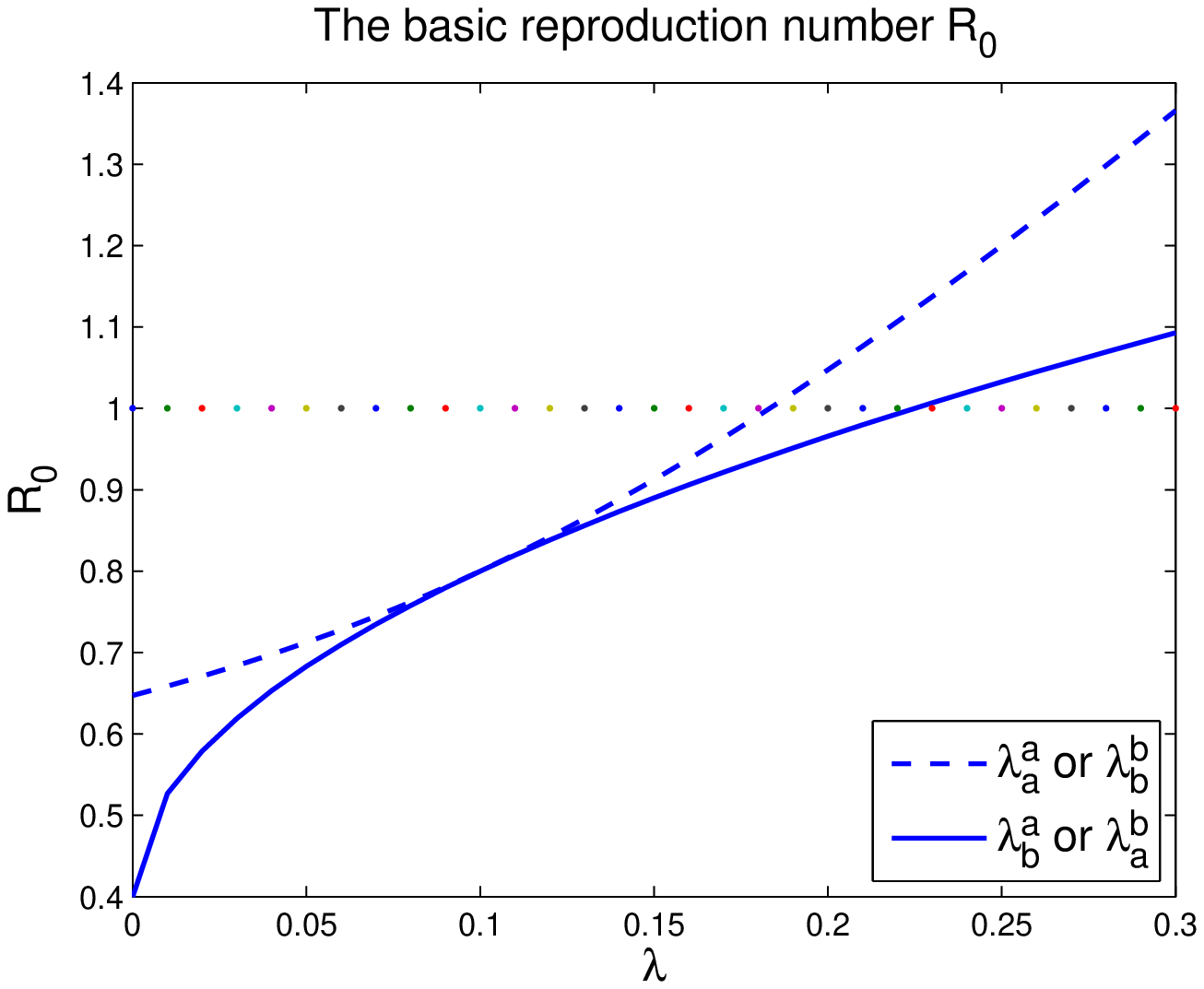}
  \end{minipage}
  }
  \subfigure{
  \begin{minipage}{7cm}
  (b)
  \centering
  \includegraphics[width=6cm]{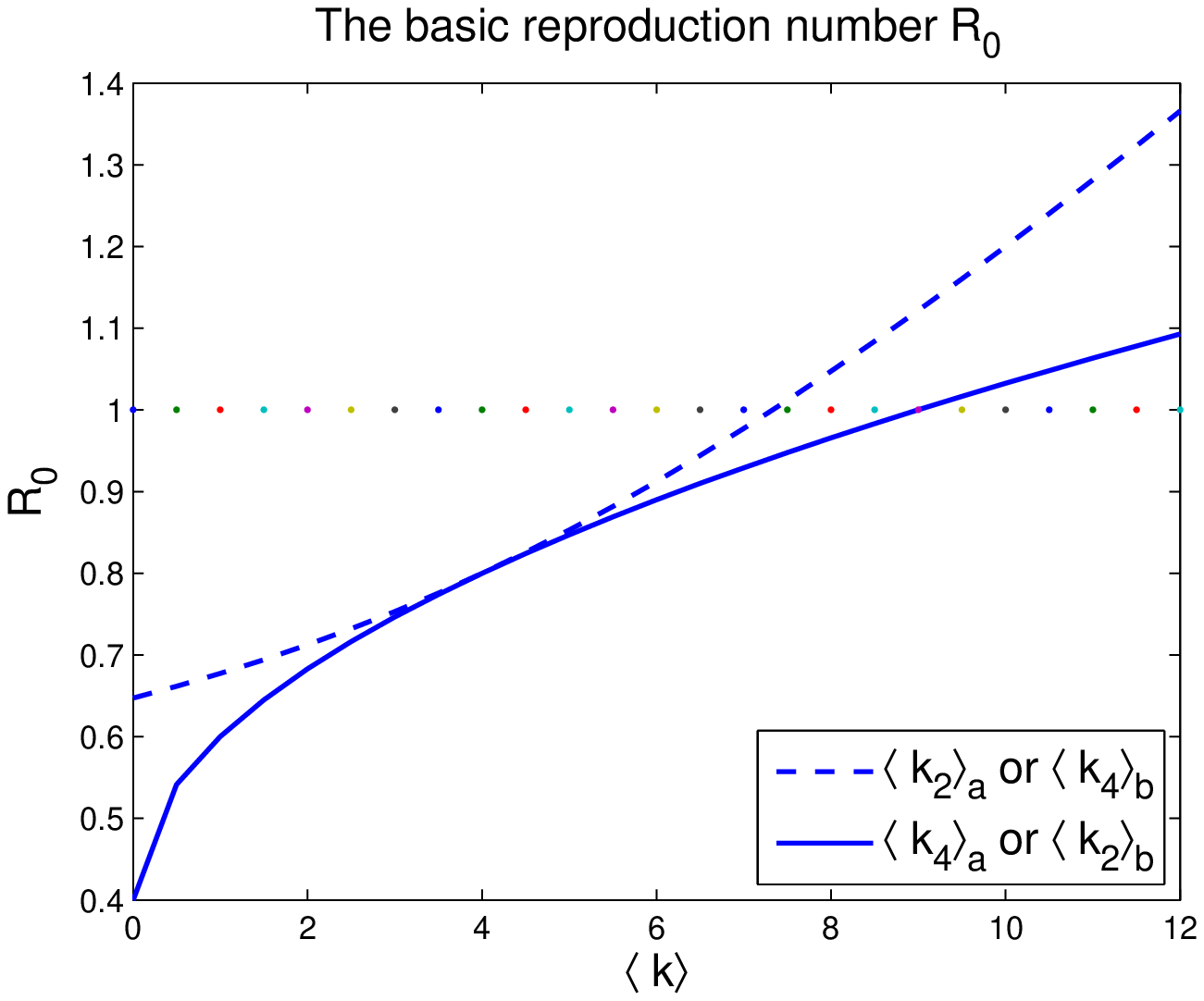}
  \end{minipage}
  }
  \caption{The effects of the infectious rate and mean degree on the basic reproduction number $R_0$.}\label{fig-2}
\end{figure}

Figure~\ref{fig-3} are also the influence graph of infectious rate
and mean degree on the basic reproduction number. Compared with Figure~\ref{fig-2},
Figure~\ref{fig-3} contain more information. The $R_0$ increases with
the increase of infectious rate or mean degree.

\begin{figure}[!h]
  \centering
  \subfigure{
  \begin{minipage}{7cm}
  (a)
  \centering
  \includegraphics[width=6cm]{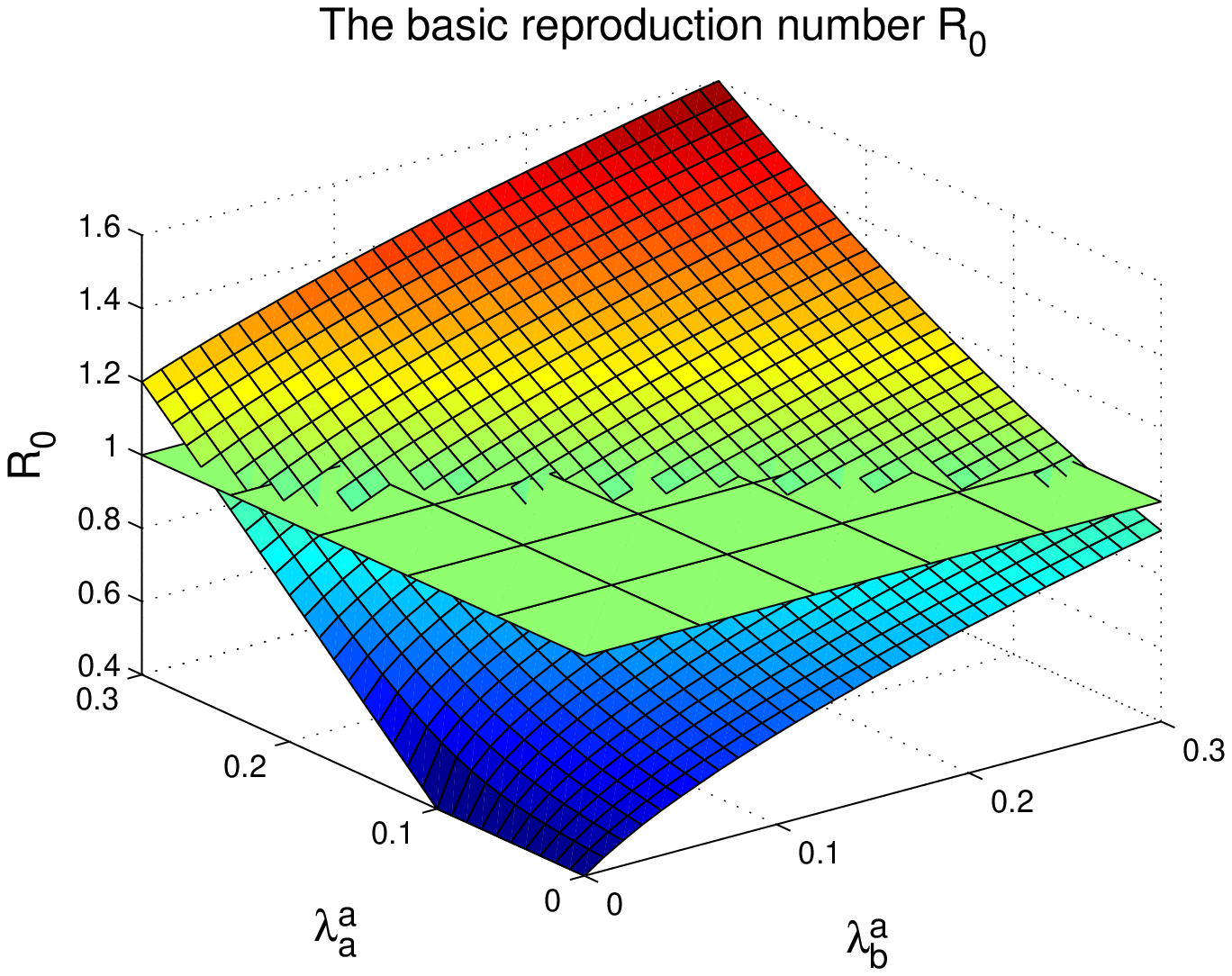}
  \end{minipage}
  }
  \subfigure{
  \begin{minipage}{7cm}
  (b)
  \centering
  \includegraphics[width=6cm]{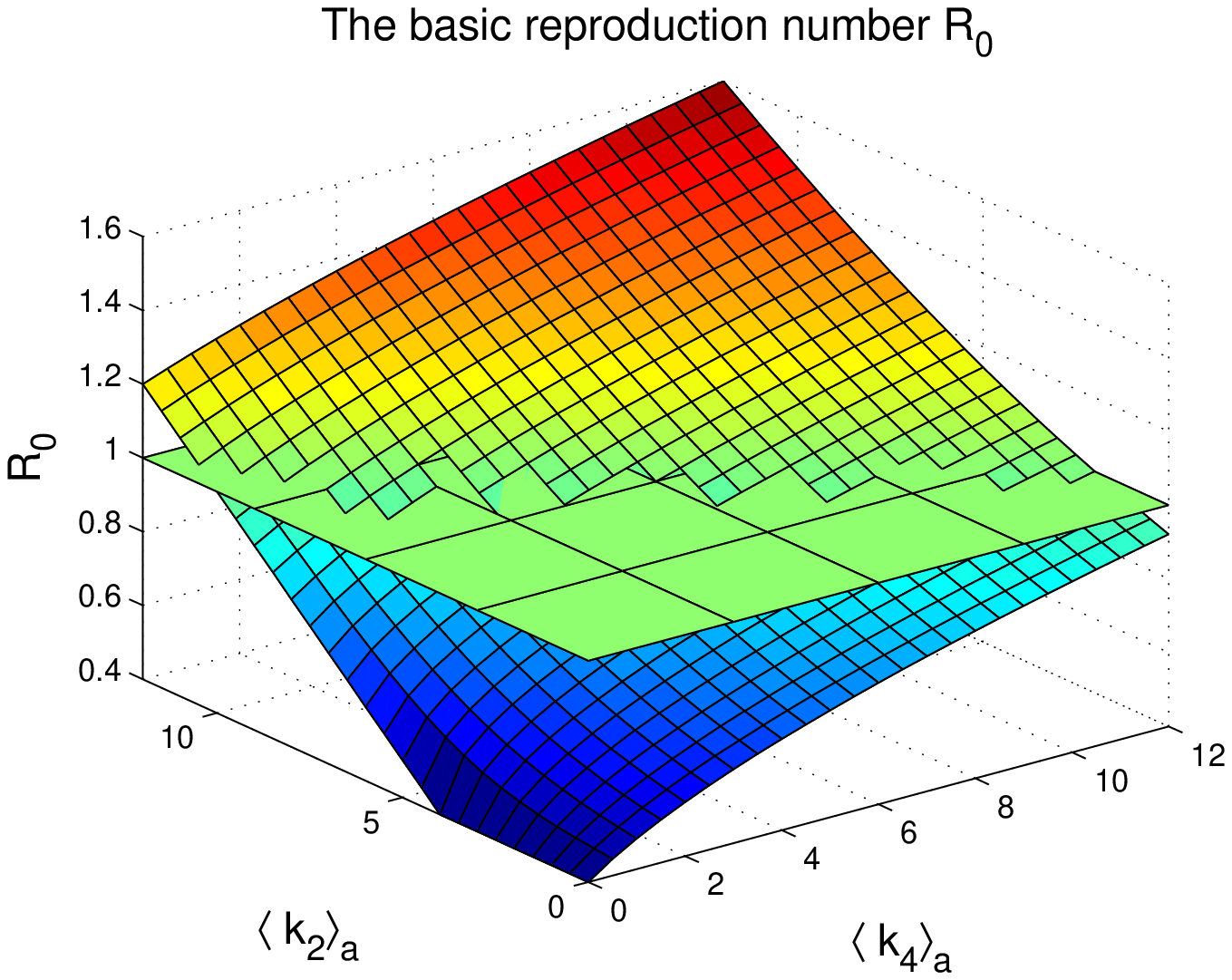}
  \end{minipage}
  }
  \caption{The effects of the infection rate on the basic reproduction number $R_0$.}\label{fig-3}
\end{figure}

Here we constructed two interconnected directed networks. For the first one,
the node number of subnetwork $A$ and $B$ are both 5000, and the sub-degree
$k_i$ of joint degree $(k_1,k_2,k_3,k_4)$ satisfies the Poisson distribution.
This network is a generalized ER network, so we denote this network by 'ER'.
For the second one, the node numbers are
both 5000, the same as above, however, the sub-degrees satisfy the
Power-law distribution with exponent 3. This network is a special scale-free
network, so we denote this network by 'SF'.
For the two networks, the same is the joint degree of every node is independent,
that is to say, the relation (\ref{eq-8}) is satisfied. And the mean degrees
both are $\langle k_1\rangle_a=\langle k_2\rangle_a=4$, $\langle k_3\rangle_a=\langle k_2\rangle_b=6$, $\langle k_4\rangle_a=\langle k_1\rangle_b=8$, $\langle k_3\rangle_b=\langle k_4\rangle_b=10$.

\begin{figure}[!ht]
  \centering
  \subfigure{
  \begin{minipage}{7cm}
  (a)
  \centering
  \includegraphics[width=6cm]{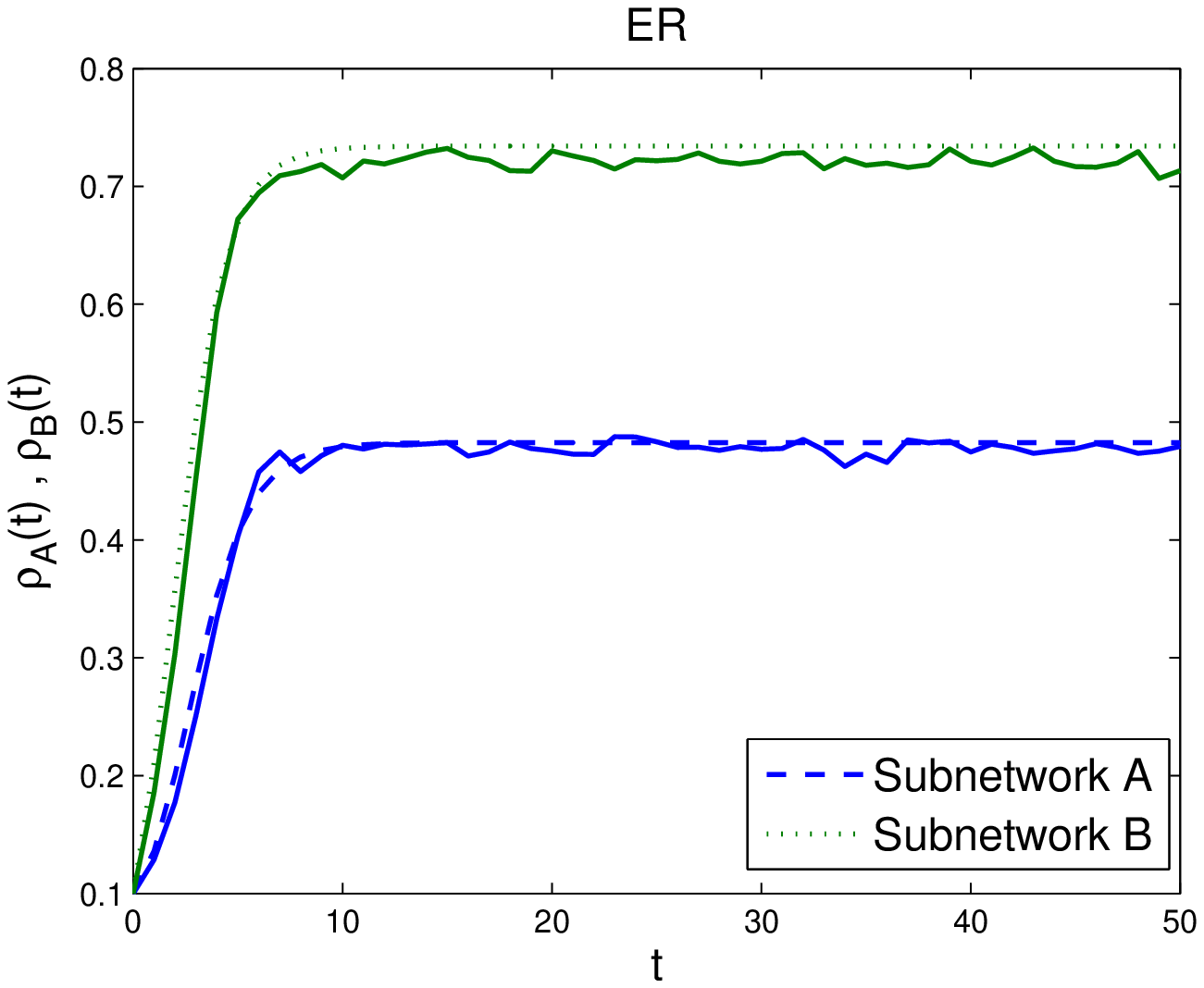}
  \end{minipage}
  }
  \subfigure{
  \begin{minipage}{7cm}
  (b)
  \centering
  \includegraphics[width=6cm]{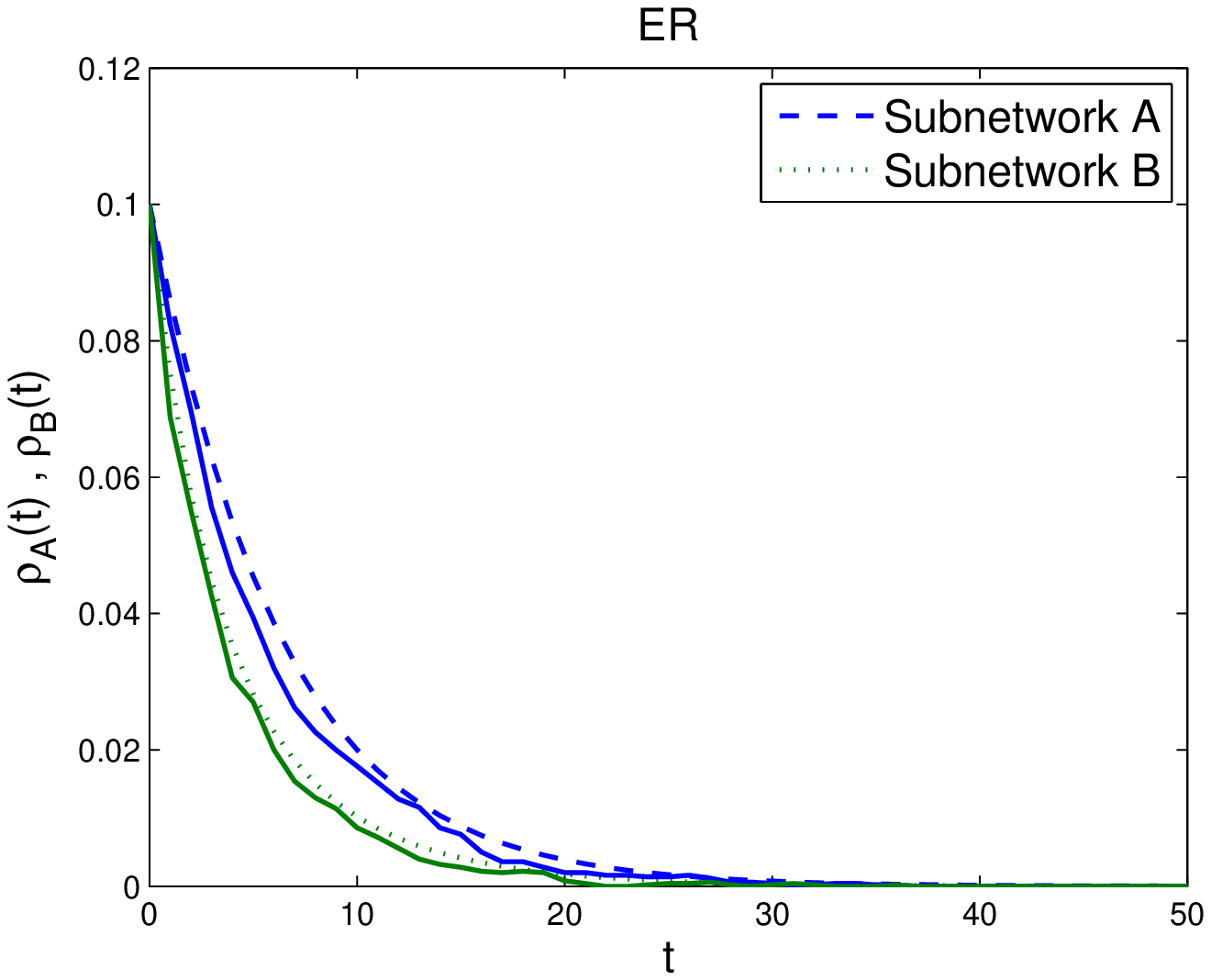}
  \end{minipage}
  }\\
  \subfigure{
  \begin{minipage}{7cm}
  (c)
  \centering
  \includegraphics[width=6cm]{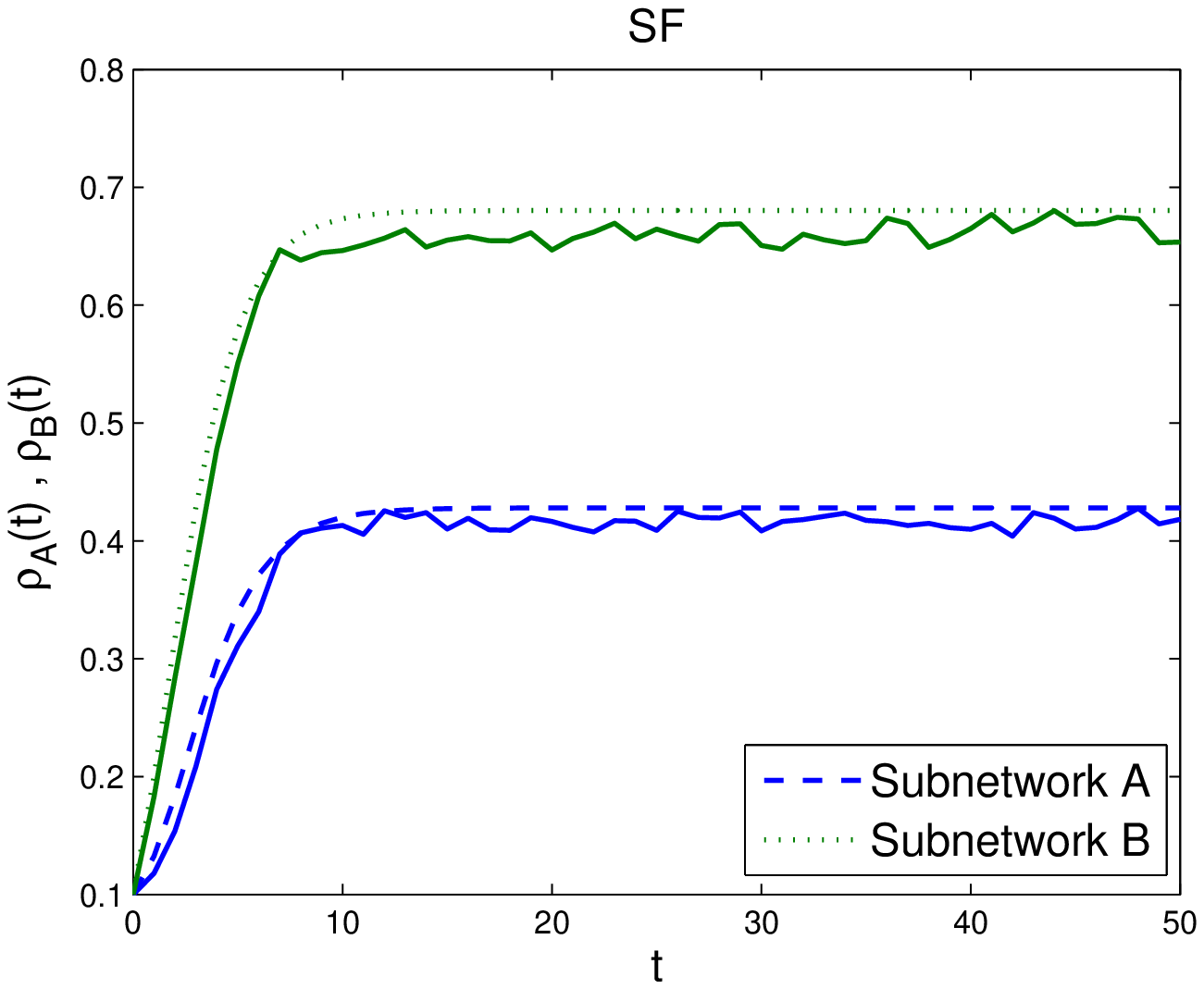}
  \end{minipage}
  }
  \subfigure{
  \begin{minipage}{7cm}
  (d)
  \centering
  \includegraphics[width=6cm]{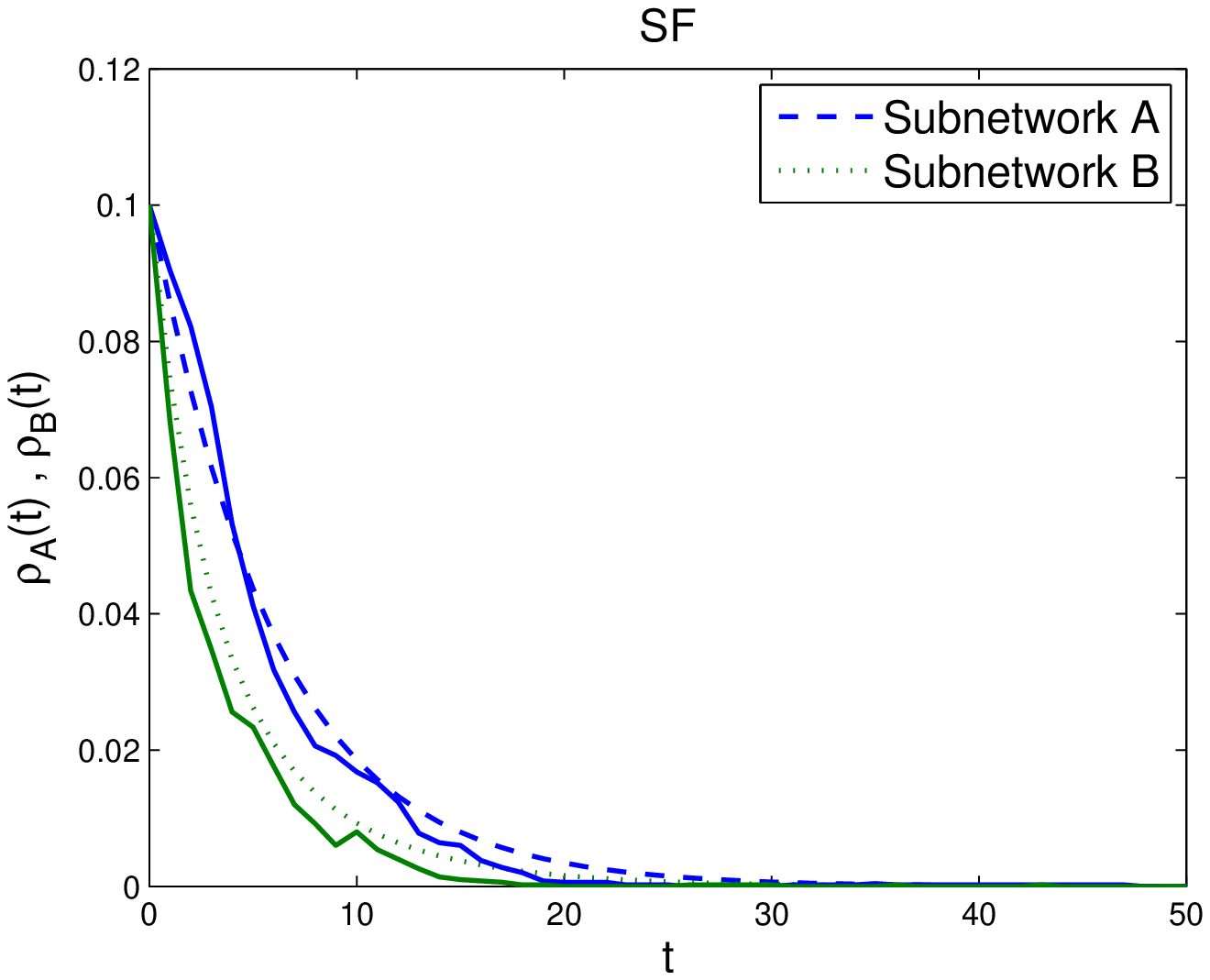}
  \end{minipage}
  }
  \caption{The total infected densities $\rho_A$ and $\rho_B$ with time $t$.}\label{fig-4}
\end{figure}

\begin{figure}[!h]
  \centering
  \subfigure{
  \begin{minipage}{7cm}
  (a)
  \centering
  \includegraphics[width=6cm]{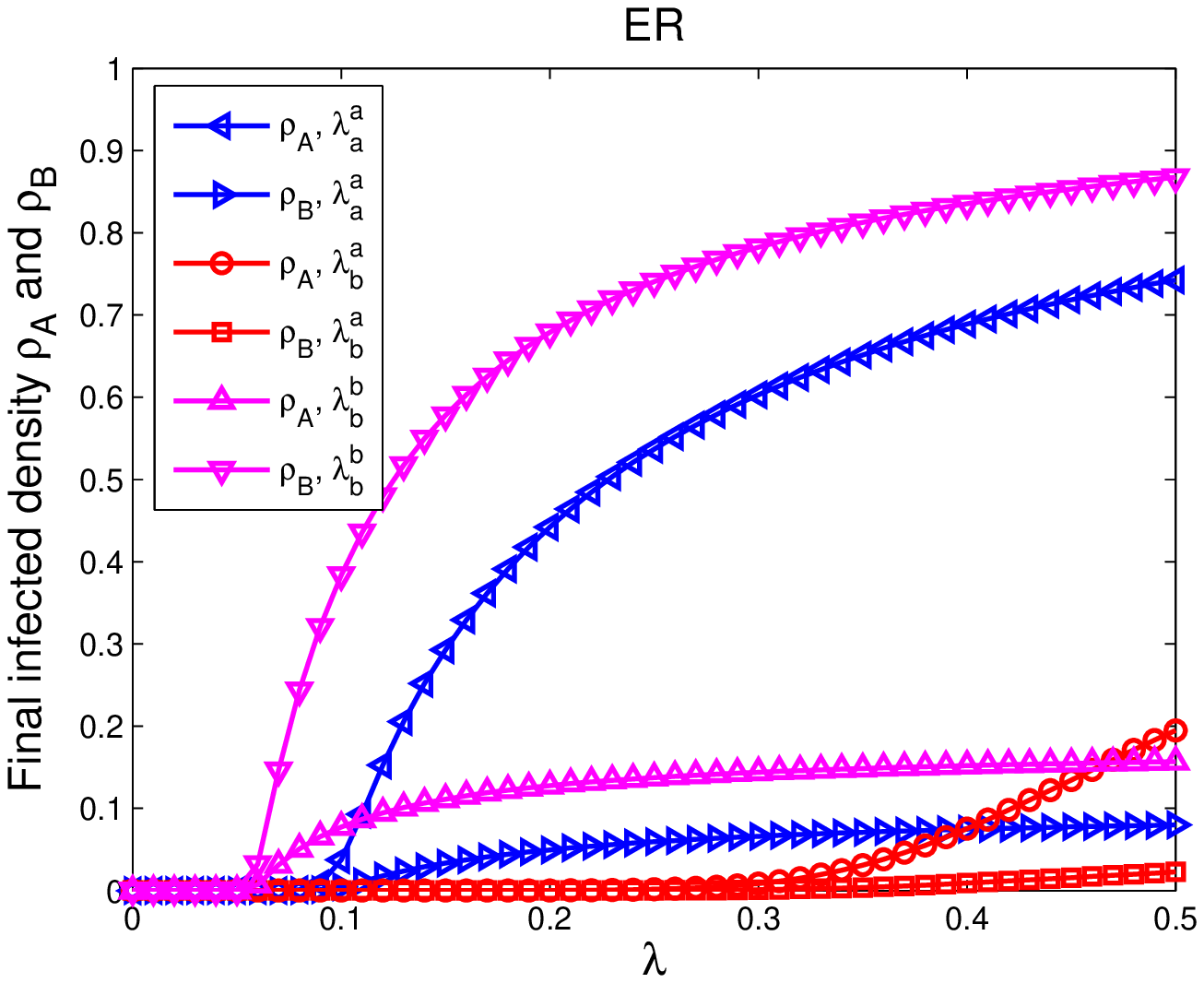}
  \end{minipage}
  }
  \subfigure{
  \begin{minipage}{7cm}
  (b)
  \centering
  \includegraphics[width=6cm]{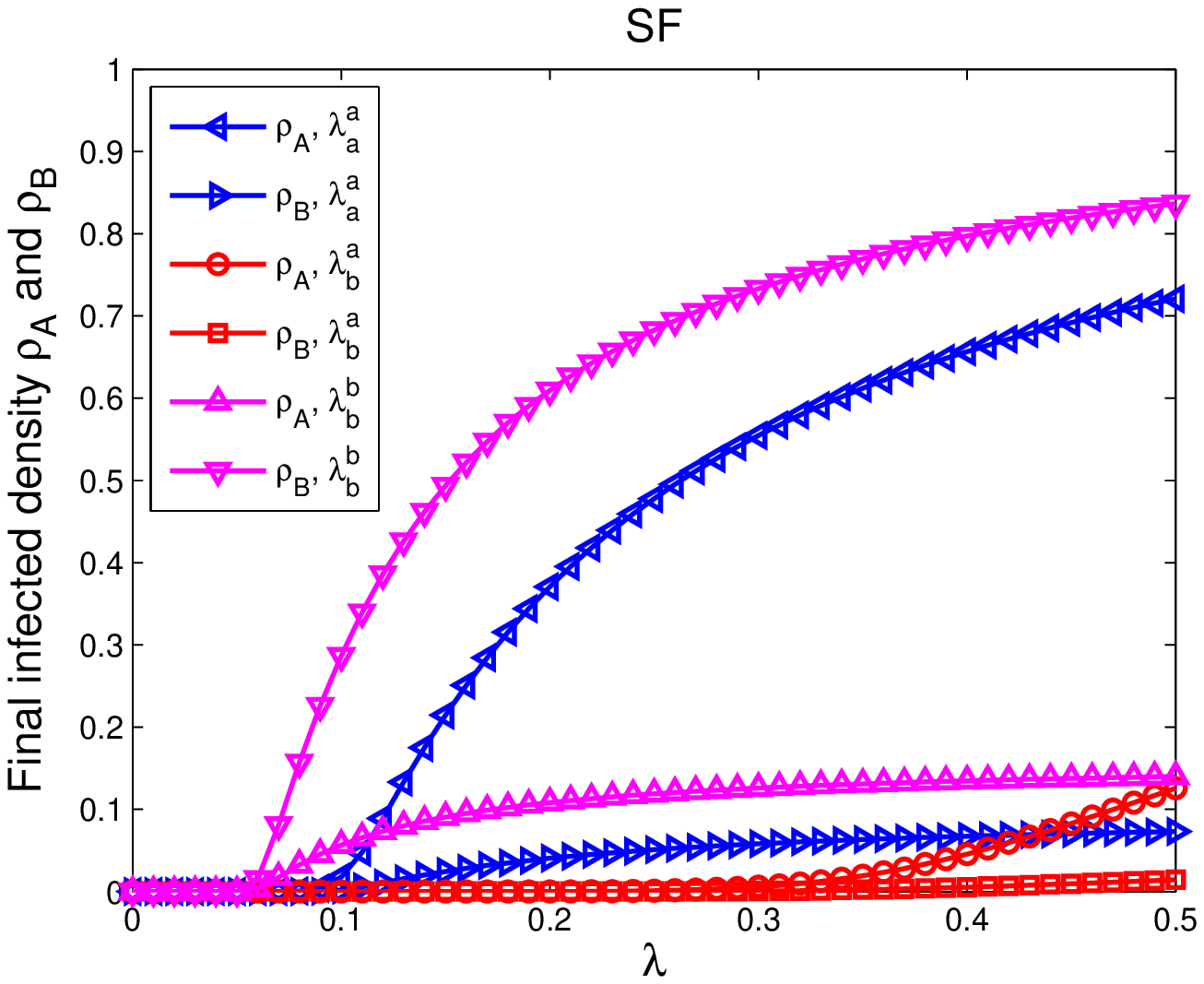}
  \end{minipage}
  }
  \caption{The final infected densities $\rho_A$ and $\rho_B$ over infectious rate $\lambda$.}\label{fig-5}
\end{figure}

We simulate the spread of disease on these two networks. The results
are plotted in Figures~\ref{fig-4}, the dashed lines are the results
of numerical simulation, and the corresponding solid lines are the single
performance of random simulation with the same parameters, respectively.
And the initial infected ratio is set to 0.1. Our simulations
are carried out with two groups of parameters. The simulations of first group is shown in (a) and (c),
where $\lambda_a^a=0.1$, $\lambda_b^a=0.05$, $\lambda_a^b=0.05$,
$\lambda_b^b=0.1$, $\mu_a=0.5$ and $\mu_b=0.4$, and the corresponding
$R_0$ is approximately equal to 3.617 ($>1$). The second is shown in (b) and (d),
where $\lambda_a^a=0.05$, $\lambda_b^a=0.01$, $\lambda_a^b=0.01$, $\lambda_b^b=0.05$, $\mu_a=0.4$ and $\mu_b=0.8$, the corresponding $R_0$ is equal to 0.725 ($<1$). We can see that there will exist an endemic when $R>1$,
otherwise, the disease will die out, and this is consistent with our theoretical
results in previous section. From these figures we obtain that the infected density of subnetwork
$A$ has the same  tendency with the one of subnetwork $B$, either becoming
an endemic or die out, but the specific densities of subnetworks may be
different. It is easy to understand that the two subnetworks are interconnected
each other, but the parameters, $\lambda$ and $\mu $, may be different.
In Figure~\ref{fig-4}, compare (a) with (c) and (b) with (d), we may find
that the tendency of infected densities in (c) and (d) are similar to
the ones in (a) and (b), respectively, though the sub-degrees of nodes in 'SF' satisfy
the power-law distribution. Because of the independence of the sub-degree,
for the nodes with large out-degree, its in-degree is not necessarily large.
Similarly, for the nodes with large in-degree, its out-degree is not necessarily large neither.
The independence of each component of joint degree reduce the effect of
power-law distribution on disease spreading, which is different from the
situation in undirected network.


Figure~\ref{fig-5} are the final infected densities with infectious rate $\lambda$,
which are the results of numerical simulations on ER and SF, respectively.
Except for the changing parameters, the other infectious rate is 0.01, and the
recovery rates are $\mu_a=0.4$, $\mu_b=0.6$. The initial infected rate is 0.1.
Though the structure of network SF is different from the one of network ER,
Figure~\ref{fig-5}(b) is similar to Figure~\ref{fig-5}(a). This once again shown
that the independence of sub-degree reduce the influence of structure on disease spread.


In addition, in Figure~\ref{fig-5}, it is shown that when the infectious rate
$\lambda$ is less than the critical value, the final infected density is 0.
When the infectious rate $\lambda$ is greater than the critical value, the final infected
density increases gradually with the increase of infectious rate. Comparing with
the inter-infectious rate, the intra-infectious rate has a greater impact on the
final infected density. What's more, the intra-infectious rate has a far greater
impact on the corresponding subnetwork than another subnetwork. For example, with
the increase of the intra-infectious rate $\lambda_a^a$, the final infected density
of subnetwork $A$ $\rho_A$ is greater than another final infected density $\rho_B$.
Although the impact of inter-infectious rate on the final infected density is relatively small,
it has the same effect as intra-infectious rate. For example, with the increase of
$\rho_b^a$, the infectious rate of subnetwork $B$ on subnetwork $A$, the final infected
density $\rho_A$ is greater than $\rho_B$.

\section{Conclusions and discussions}
\label{sec-5}

Motivated by the interaction between the actual systems and the direction
of information dissemination, we establish an SIS model in an interconnected
directed network for studying the epidemic spread. This network is the generalization
of undirected networks and bipartite networks, in other words, this interconnected
directed network can be transformed into these two kind of networks in special
condition. We theoretically analyze the model, and obtain the basic reproduction
number $R_0$, which is also a generalized threshold value. We prove that the disease
will become endemic if the $R_0$ greater than one, otherwise, the disease
will die out. We also give a condition for epidemic prevalence only on a single subnetwork.
By numerical analysis we find that the independence of joint
degree can greatly reduce the effect of heterogeneity of degree on disease
spread.

\subsection*{Acknowledgement}
This work was jointly supported by the NSFC grants under Grant Nos. 11331009 and 11572181.

\end{document}